\documentclass[sigconf]{acmart}

\AtBeginDocument{%
  }

\setcopyright{acmcopyright}
\copyrightyear{2026}
\acmYear{2026}
\acmConference[Conference '26]{Conference '26}{Location}{}
\acmISBN{978-1-4503-XXXX-X/26/06}
\acmDOI{10.1145/XXXXXX.XXXXXX}

\usepackage{tikz}
\usetikzlibrary{positioning,arrows.meta,calc,fit,patterns,decorations.pathreplacing}
\usepackage{pgfplots}
\usepgfplotslibrary{fillbetween}
\pgfplotsset{compat=1.18}
\usepackage{booktabs}

\usepackage{amsmath,amssymb}
\usepackage{amsthm}
\newtheorem{proposition}{Proposition}
\usepackage{algorithm}
\usepackage{algorithmic}
\usepackage{multirow}
\usepackage{xcolor}
\usepackage{subcaption}
\usepackage{graphicx}
\graphicspath{{figures/}}
\definecolor{pathblue}{HTML}{2E78C2}
\definecolor{pathgreen}{HTML}{3A8F3C}
\definecolor{gateorange}{HTML}{D4820E}
\definecolor{lightblue}{HTML}{DCEEFB}
\definecolor{lightgreen}{HTML}{DFF0D8}
\definecolor{lightorange}{HTML}{FEF3CD}

\begin{document}

\title{DeRes: Decoupling Residual Stability and Adaptivity for Scalable CTR Prediction}

\author{Wenzhuo Cheng}
\authornote{Both authors contributed equally to this research.}
\affiliation{\country{}}

\author{Shipeng Nie}
\authornotemark[1]
\affiliation{\country{}}

\author{Qixin Guo}
\affiliation{\country{}}

\author{Xuefeng Sun}
\affiliation{\country{}}

\author{Jianguo Lou}
\affiliation{\country{}}

\author{Zhengwei Zheng}
\authornote{Corresponding author.}
\affiliation{\country{}}

\renewcommand{\shortauthors}{Wenzhuo Cheng et al.}

\begin{abstract}
Transformer-based CTR models face a growing bottleneck at the residual connection: under Pre-Norm, early user-interest signals are diluted layer by layer; the identity skip cannot forget stale interests; and each layer sees only its immediate predecessor, losing long-range cross-layer dependencies. Recent attention-based residual variants (AttnRes) address parts of this in language models, but they drop the protective identity skip and have not been tried in recommendation. Drawing on Dual Path Networks (DPN) and the HORNN view of residuals, we present \textbf{DeRes}, which routes each layer through two paths in parallel---an \textit{Identity residual path} that preserves first-order feature reuse and gradient flow, and a \textit{Block Attention Residual path} that allows each layer to attend over compressed outputs from all earlier blocks for high-order recall. A vector-wise gate determines, for each hidden dimension, the relative weight assigned to each path. We further propose \textit{Pointwise AttnRes}, replacing the Softmax in the cross-layer attention with SiLU so that multiple past blocks can be activated simultaneously and irrelevant ones can receive negative (forgetting) weights---a setting better aligned with CTR's parallel multi-interest patterns. On a large-scale industrial dataset (331M interactions from a major social-media platform), Criteo (45M), and Avazu (40M), DeRes outperforms twelve baselines including OneTrans, TokenMixer-Large, UniMixer, mHC, and AttnRes, achieving up to $+0.32\%$ AUC at under $5\%$ additional FLOPs. Beyond a single operating point, DeRes fits a markedly steeper compute--AUC scaling law ($\gamma{=}0.118$ vs.\ $0.071$ for OneTrans, a $1.66\times$ gap), such that an 8-layer DeRes matches the performance of a 16-layer OneTrans---approximately $2\times$ compute saving at equivalent AUC. Ablations confirm that the dual-path design outperforms either single path, the Identity skip outperforms learnable residuals, and SiLU outperforms Softmax.
\end{abstract}



\maketitle

\section{Introduction}

Click-through rate (CTR) prediction sits at the center of search, social feeds, and e-commerce ranking, where small accuracy gains translate to real revenue~\cite{guo2017deepfm,zhou2018din}. The field has moved from shallow feature-interaction models such as FM~\cite{rendle2010fm}, DeepFM~\cite{guo2017deepfm} and DCN-v2~\cite{wang2021dcnv2} to deep Transformer-based architectures~\cite{vaswani2017attention} that jointly model behavior sequences and feature interactions~\cite{chen2019bst,zhai2024hstu}. Recent industrial ranking systems have likewise scaled up backbone depth and width: OneTrans~\cite{chen2025onetrans} at TikTok stacks deep causal Transformers, while TokenMixer-Large~\cite{liu2026tokenmixerlarge} at ByteDance and UniMixer~\cite{wang2026unimixer} at Kuaishou pursue parallel token-mixing designs. All three report clear online gains and motivate a closer look at how depth is actually used inside such models.

As these models grow, a less obvious component becomes a bottleneck: the \textit{residual connection}. The standard form $\mathbf{x}_l = \mathbf{x}_{l-1} + f_l(\mathbf{x}_{l-1})$~\cite{he2016resnet} keeps gradients alive, but it carries three side effects we observe in CTR Transformers:

\begin{enumerate}
    \item \textbf{Pre-Norm signal dilution.} Under Pre-Norm~\cite{xiong2020prenorm}, layer outputs are added to normalized residuals, so early-layer features are gradually washed out. In CTR, this includes the layers that encode the user's earliest click signals.
    \item \textbf{No way to forget.} The identity skip preserves everything with weight $1.0$. When user interests drift, stale signals stay in the residual stream and interfere with the prediction.
    \item \textbf{Single-layer view.} Each layer can only read $\mathbf{x}_{l-1}$, so any pattern that lives across non-adjacent layers is invisible.
\end{enumerate}

Several recent works rethink the inter-layer connection in language models. DenseFormer~\cite{pagliardini2024denseformer} uses static depth-weighted averaging---48 layers match a 72-layer Transformer. Attention Residuals (AttnRes)~\cite{yang2026attnres} replaces fixed weights with Softmax cross-layer attention and reports $\sim$25\% token-equivalent savings on Kimi-48B at $<5\%$ overhead. Hyper-Connections (HC)~\cite{zhu2024hc} and the Sinkhorn-projected mHC~\cite{zhu2025mhc} expand the residual stream itself; MUDDFormer~\cite{li2025muddformer} splits Q/K/V/R into four dynamic dense connections. None of them has been tried on CTR data, and---more importantly for our purposes---none combines dynamic cross-layer attention with an explicit identity protection path.

CTR Transformers are not just shallow language models. They are typically 4--12 layers; the input is a heterogeneous mix of categorical, numerical, and sequential features rather than a token stream; and a single prediction often has to encode several interests in parallel. Zhang et al.~\cite{zhang2025fat} show via Rademacher complexity that the relevant scaling lever is the number of feature fields $F$, not the sequence length $n$. Wukong~\cite{zhang2024wukong} pushes this empirically: a dual-path design (FMB for high-order, LCB for low-order) keeps scaling across two orders of magnitude on 146B examples while single-path baselines saturate.

Building on Dual Path Networks (DPN)~\cite{chen2017dpn}---which combined ResNet's first-order reuse with DenseNet's~\cite{huang2017densenet} high-order exploration under the HORNN view of recurrent depth~\cite{li2015hornn}---we propose \textbf{DeRes}, a dual-path inter-layer connector for CTR Transformers with three components:

\begin{itemize}
    \item an \textbf{Identity residual path} with fixed identity weights, keeping gradient flow and first-order reuse intact;
    \item a \textbf{Block Attention Residual path} that attends, layer by layer, over compressed block summaries of every earlier layer for high-order recall;
    \item a \textbf{vector-wise gate} that lets each hidden dimension decide how much weight to give either path.
\end{itemize}

We then propose \textbf{Pointwise AttnRes}, swapping the Softmax in the cross-layer attention for a SiLU activation~\cite{ramachandran2017swish}. This breaks the zero-sum competition that Softmax imposes: several earlier blocks can be activated at the same time, which matches CTR's parallel multi-interest pattern, and the negative branch of SiLU lets the layer mute irrelevant history---following the same intuition behind HSTU's removal of intra-layer Softmax~\cite{zhai2024hstu}.

The design rests on a single principle: \textit{residual stability and cross-layer adaptivity should not share one operator.} Plain residuals are stable but static; DenseFormer and AttnRes are adaptive but drop the identity protection; learnable residual matrices add mixing capacity but accumulate spectral attenuation across depth. DeRes splits the two roles---the identity path handles low-order and shallow signals, the AttnRes path retrieves higher-order ones from earlier layers---which fits the heterogeneous structure of CTR inputs and, as our results show, optimizes more cleanly.

The contributions of the paper are:
\begin{enumerate}
    \item We introduce Attention Residuals into CTR Transformers for the first time and show that the resulting gain is comparable to that observed in language models, and in relative terms larger given the shallower depth and heterogeneous inputs characteristic of CTR data.
    \item We introduce DeRes, a DPN-style dual-path connector that pairs an identity residual with a block-level attention residual, grounded in the HORNN view.
    \item We introduce Pointwise AttnRes with SiLU, a CTR-specific replacement for Softmax cross-layer attention that supports non-competitive multi-interest encoding.
    \item We evaluate against twelve baselines on a large-scale industrial dataset, Criteo, and Avazu---six feature / sequence models and six inter-layer connectors (DenseFormer, mHC, AttnRes, OneTrans, TokenMixer-L, UniMixer)---with ablations, a scaling-law fit, and per-segment analyses.
\end{enumerate}

\section{Related Work}

\subsection{CTR Prediction Models}

CTR prediction has evolved from feature engineering and shallow models to deep neural architectures. Factorization Machines (FM)~\cite{rendle2010fm} model pairwise feature interactions. Wide\&Deep~\cite{cheng2016wide} combines memorization and generalization. DeepFM~\cite{guo2017deepfm} integrates FM with deep networks; xDeepFM~\cite{lian2018xdeepfm} adds a compressed interaction network for explicit high-order feature crossing. DCN-v2~\cite{wang2021dcnv2} models explicit cross features via cross networks. AutoInt~\cite{song2019autoint} and FiBiNET~\cite{huang2019fibinet} use attention to learn feature interactions automatically.

For user behavior modeling, early RNN-based approaches such as GRU4Rec~\cite{hidasi2016gru4rec} encoded sessions with recurrent units; DIN~\cite{zhou2018din} and DIEN~\cite{zhou2019dien} then introduced attention-weighted interest extraction for CTR. Transformer-based approaches~\cite{vaswani2017attention} brought self-attention to sequential CTR: BST~\cite{chen2019bst} applies Transformer to behavior sequences; SASRec~\cite{kang2018sasrec} uses causal self-attention for next-item prediction. HSTU~\cite{zhai2024hstu} demonstrated trillion-parameter sequential Transformers for recommendation.

Recent industrial work has focused on scaling CTR Transformers. OneTrans~\cite{chen2025onetrans} unifies heterogeneous features into a single Transformer with standard residual connections, reporting +0.748\% DAU online. TokenMixer-Large~\cite{liu2026tokenmixerlarge} introduces Inter-Residual connections and MoE, scaling to 15B parameters. UniMixer~\cite{wang2026unimixer} unifies attention, token-mixing, and FM into a single parameterization with Sinkhorn-constrained mixing, achieving +15\% user retention. EST~\cite{chen2026est} proposes single-directional cross-attention for efficiency. All four use either standard or static inter-layer connections; none explores dynamic cross-layer attention---which is the gap DeRes targets.

A summary of inter-layer connection methods across the full evolution is provided in Table~\ref{tab:ilc_comparison}.

\begin{table}[t]
\caption{Taxonomy of inter-layer connection methods. DeRes is the only design that combines dynamic cross-layer attention with identity-residual protection, and the first such design validated on CTR prediction. FLOPs overhead values are as reported in the original paper / domain; see Table~\ref{tab:efficiency} for measured CTR overhead under our unified protocol.}
\Description{Comparison table of inter-layer connection methods showing weight type, residual protection, cross-layer access, original domain, FLOPs overhead, and CTR validation status.}
\label{tab:ilc_comparison}
\centering
\small
\setlength{\tabcolsep}{2.5pt}
\begin{tabular}{lcccccc}
\toprule
\textbf{Method} & \textbf{Weight} & \textbf{Res.} & \textbf{Cross-} & \textbf{Domain} & \textbf{+FLOPs} & \textbf{CTR} \\
& \textbf{type} & \textbf{prot.} & \textbf{layer} & & & \textbf{val.} \\
\midrule
ResNet & fixed 1.0 & \checkmark & $l$-1 only & General & 0\% & --- \\
DenseNet & concat & $\times$ & all prev. & Vision & high & $\times$ \\
DPN & fixed+concat & \checkmark & all prev. & Vision & moderate & $\times$ \\
DenseFormer & static & $\times$ & all prev. & LLM & $\sim$5\% & $\times$ \\
HC/DHC & learnable & $\times$ & $n$-way & LLM & $\sim$3\% & $\times$ \\
AttnRes & Softmax & $\times$ & all prev. & LLM & $<$5\% & $\times$ \\
mHC & doubly stoch. & $\times$ & $n$-way & LLM & 6.7\% & $\times$ \\
MUDDFormer & MLP & $\times$ & all prev. & LLM & 0.4\% & $\times$ \\
ANCRe & softmax & $\times$ & all prev. & LLM & $<$1\% & $\times$ \\
\midrule
\textbf{DeRes} & \textbf{Id $+$ SiLU} & \checkmark & \textbf{block} & \textbf{CTR} & \textbf{$<$5\%} & \checkmark \\
\bottomrule
\end{tabular}
\end{table}

\subsection{Inter-Layer Connections}

The standard residual connection $\mathbf{x}_l = \mathbf{x}_{l-1} + f_l(\mathbf{x}_{l-1})$~\cite{he2016resnet} has been the default inter-layer mechanism since 2015. DenseNet~\cite{huang2017densenet} concatenates all previous layer outputs, providing direct access but with linearly growing memory. DPN~\cite{chen2017dpn} combines both paradigms in parallel, justified by the HORNN framework~\cite{li2015hornn} that views ResNet as first-order and DenseNet as higher-order recurrent networks.

In the Transformer era, DenseFormer~\cite{pagliardini2024denseformer} learns \textit{static} depth-weighted averaging coefficients, making 48 layers equivalent to 72. HC~\cite{zhu2024hc} expands the residual stream to $n \times d$ dimensions with learnable $n \times n$ transition matrices. Its manifold-constrained variant mHC~\cite{zhu2025mhc} uses Sinkhorn-projected doubly stochastic matrices for stability. However, recent experiments on Qwen3-1.7B/8B (150B tokens) show that Identity residual matrices ($\mathbf{H}^{res} = \mathbf{I}$) outperform learned alternatives---doubly stochastic matrices with $|\lambda_{\min}| \approx 0.49$ cause cumulative signal collapse: $\sigma_{\min} \approx 0.49^{56} \approx 4.5 \times 10^{-18}$ over 28 layers.

AttnRes~\cite{yang2026attnres} replaces fixed residual weights with Softmax-based dynamic cross-layer attention, with clear gains on Kimi-48B. MUDDFormer~\cite{li2025muddformer} uses four independent dynamic connections for Q/K/V/R, getting 2.8B-parameter performance from 6.9B-parameter quality. ANCRe~\cite{houlsby2026ancre} shows that residual topology has exponential impact on convergence speed, and that ingoing softmax normalization is needed for stable training.

\textit{All of these methods target language models.} CTR Transformers are different in three concrete ways: heterogeneous inputs (numerical, categorical, sequential), much shallower depth (4--12 layers), and parallel multi-interest encoding instead of a single autoregressive context. To our knowledge, dynamic inter-layer connections have not previously been adapted to CTR, and DeRes does so with DPN-style dual-path protection and a CTR-specific pointwise activation.

\section{The DeRes Framework}

\subsection{Preliminaries}

\textbf{CTR Transformer.} Given a user's feature set and behavior sequence, the input is embedded as $\mathbf{X}^{(0)} \in \mathbb{R}^{n \times d}$ where $n$ is the sequence length and $d$ the hidden dimension. A stack of $L$ Transformer layers~\cite{vaswani2017attention} processes this input, with each layer $l$ consisting of multi-head self-attention (MHA) and feed-forward network (FFN) sub-layers with Pre-Norm~\cite{xiong2020prenorm} and LayerNorm~\cite{ba2016layernorm}:
\begin{equation}
\begin{aligned}
\mathbf{Y}^{(l)} &= \mathbf{X}^{(l-1)} + \text{MHA}(\text{LN}(\mathbf{X}^{(l-1)})), \\
\mathbf{X}^{(l)} &= \mathbf{Y}^{(l)} + \text{FFN}(\text{LN}(\mathbf{Y}^{(l)})).
\end{aligned}
\end{equation}
The final prediction uses the target item's representation: $\hat{y} = \sigma(\mathbf{w}^\top \mathbf{x}^{(L)}_{\text{target}})$.

\textbf{Standard Residual Limitation.} The identity skip connection $\mathbf{x}_l = \mathbf{x}_{l-1} + f_l(\mathbf{x}_{l-1})$ provides a direct gradient path but constrains information flow: layer $l$ can only access its immediate predecessor, earlier signals are progressively diluted, and once information enters the residual stream, it cannot be selectively forgotten.

\subsection{Design Goals}

A CTR-friendly inter-layer connector should satisfy four properties that the standard residual does not jointly provide.

\textbf{Stability.} The connector must keep gradients clean. CTR data is sparse and high-cardinality, so rare-feature signals are easily lost when depth-wise transformations are unstable. The identity skip remains the most reliable way to preserve this path.

\textbf{Selectivity.} The connector should determine which earlier representations matter for the current prediction. Behavior sequences contain repeated clicks, stale preferences, and noisy exposures, and a fixed skip cannot distinguish long-term interests from outdated ones.

\textbf{Parallelism.} A user often holds several interests simultaneously---electronics, apparel, sports---which must be encoded together. Forcing cross-layer weights to sum to one suppresses the secondary interests.

\textbf{Efficiency.} The connector must fit within a ranking-stage budget. Dense layer-to-layer attention scales poorly with depth, so we compress earlier layers into blocks and attend over block summaries rather than over every layer.

DeRes addresses these requirements through complementary components: the Identity path provides stability, the AttnRes path provides selectivity, the Pointwise activation provides parallelism, and block compression keeps the cost down.

\subsection{DeRes Architecture}

\begin{figure*}[!t]
\centering
\includegraphics[width=\textwidth]{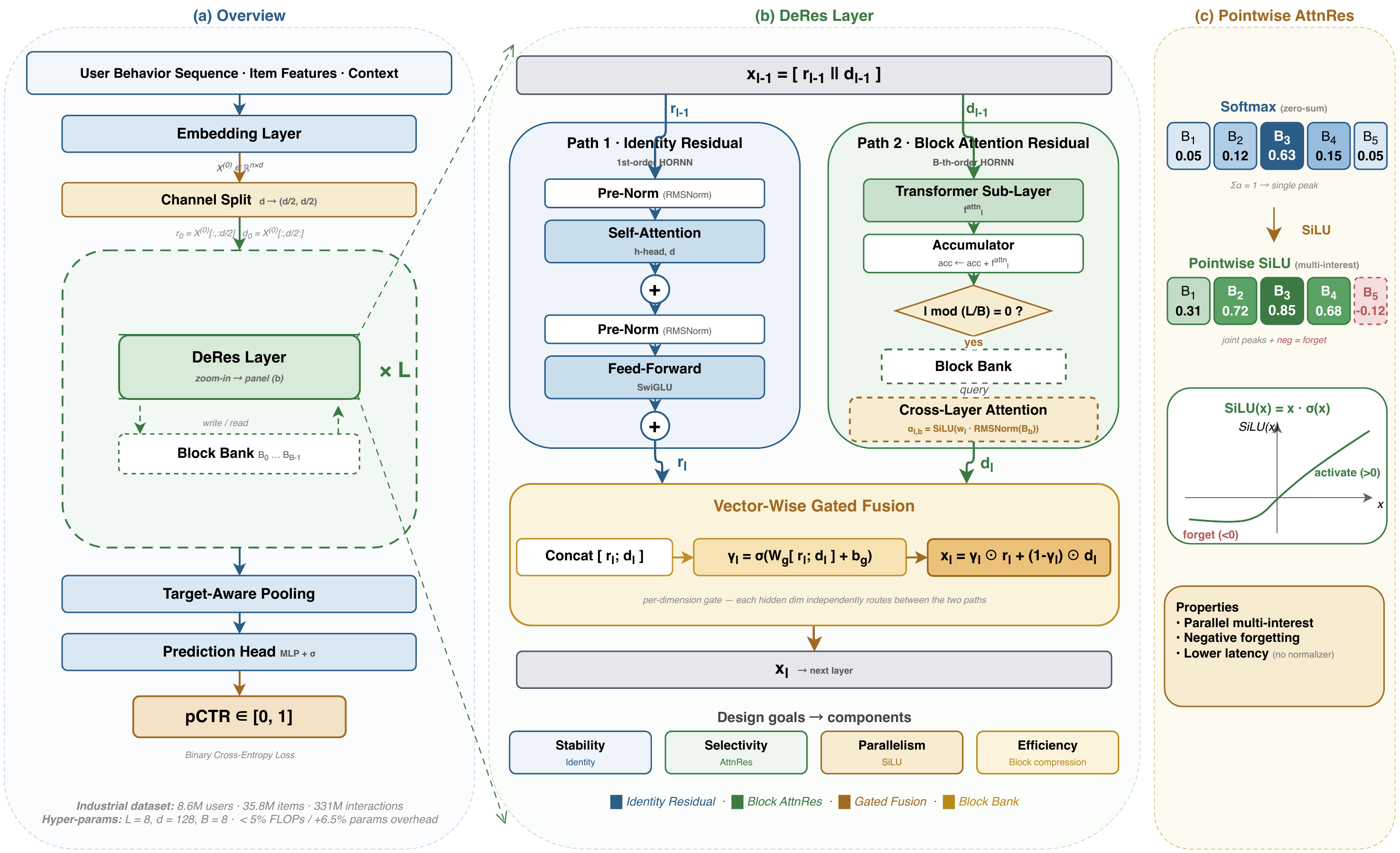}
\caption{DeRes architecture. \textbf{(a)}~CTR pipeline: embedding, channel split ($d{\to}d/2{+}d/2$), $L$ DeRes layers with a shared block bank, target-aware pooling, and prediction head. \textbf{(b)}~A DeRes layer: Path~1 (blue, identity residual, first-order HORNN) and Path~2 (green, block attention residual, higher-order HORNN); Path~2 accumulates layer outputs and writes to the block bank when $l \bmod (L/B){=}0$, then attends over the bank with $\alpha_{l,b}{=}\mathrm{SiLU}(\mathbf{w}_l{\cdot}\mathrm{RMSNorm}(\mathbf{B}_b))$. The two paths fuse via a vector-wise gate $\boldsymbol{\gamma}_l$. The four design goals (stability, selectivity, parallelism, efficiency) map to Identity, AttnRes, SiLU, and block compression respectively. \textbf{(c)}~Pointwise AttnRes replaces Softmax (zero-sum, single peak) with SiLU for parallel multi-interest activation and negative-valued soft forgetting.}
\Description{Three-panel architecture figure spanning both columns. Panel (a) shows the end-to-end pipeline. Panel (b) shows one DeRes layer with the identity residual path on the left in blue and the block attention residual path on the right in green, fused by a vector-wise gate in amber. Panel (c) contrasts Softmax (zero-sum) against Pointwise SiLU (multi-interest with forgetting).}
\label{fig:architecture}
\end{figure*}

Figure~\ref{fig:architecture} summarizes the design: panel~(a) places DeRes in the full CTR pipeline, panel~(b) zooms into a single layer (Path~1, Path~2 with the per-block accumulator and bank-write condition $l\bmod(L/B){=}0$, and the vector-wise gate), and panel~(c) contrasts Softmax with Pointwise SiLU cross-layer attention. The Algorithm~\ref{alg:dualres} listing later in this section makes the bank-write and read steps explicit.

\textbf{Path~1: Identity Residual (Feature Reuse).} The first path maintains the standard residual connection with identity skip weights:
\begin{equation}
\mathbf{r}_l = \mathbf{r}_{l-1} + f^{\text{res}}_l(\mathbf{r}_{l-1})
\end{equation}
where $f^{\text{res}}_l$ denotes the Transformer sub-layer (MHA followed by FFN). Under the HORNN framework~\cite{li2015hornn,chen2017dpn}, this corresponds to a \textit{first-order} recurrence that reuses features from the immediate predecessor. We use a fixed identity matrix rather than a learnable residual matrix, both for the theoretical reason in Proposition~\ref{prop:stability} (a learnable residual matrix with smallest singular value below $1$ accumulates exponential attenuation across depth) and because Qwen3 experiments report that $\mathbf{H}^{res}{=}\mathbf{I}$ ends up beating the learned alternatives in practice.

\textbf{Path~2: Block Attention Residual (Cross-Layer Exploration).} The second path dynamically attends to outputs from all previous layers, grouped into $B$ blocks for efficiency. Given $L$ layers divided into $B$ blocks, the $b$-th block output is the sum of layer outputs within that block:
\begin{equation}
\mathbf{B}_b = \sum_{l \in \text{Block}_b} \mathbf{g}_l(\mathbf{d}_l)
\end{equation}
where $\mathbf{g}_l$ is the Transformer sub-layer for the AttnRes path. Let $b(l)$ denote the index of the most recent block written before (or at) layer $l$. The cross-layer attention weight for block $b$ at layer $l$ is computed as (RMSNorm~\cite{zhang2019rmsnorm} is applied to each block summary):
\begin{equation}
\alpha_{l,b} = \frac{\exp(\mathbf{w}_l \cdot \text{RMSNorm}(\mathbf{B}_b))}{\sum_{b'=0}^{b(l)} \exp(\mathbf{w}_l \cdot \text{RMSNorm}(\mathbf{B}_{b'}))}
\label{eq:attnres_softmax}
\end{equation}
where $\mathbf{w}_l \in \mathbb{R}^{d/2}$ is a static trainable query vector for layer $l$ (matching the $d/2$-dimensional per-path channel split, see Algorithm~\ref{alg:dualres}). The AttnRes path output is:
\begin{equation}
\mathbf{d}_l = \sum_{b=0}^{b(l)} \alpha_{l,b} \cdot \mathbf{B}_b
\end{equation}

Under the HORNN framework, this path is a \textit{higher-order} recurrence: later layers can read directly from compressed summaries of any earlier stage rather than only from the immediate predecessor.

\textbf{Gated Fusion.} The outputs of both paths are fused via a vector-wise sigmoid gate:
\begin{equation}
\boldsymbol{\gamma}_l = \sigma(\mathbf{W}_g [\mathbf{r}_l ; \mathbf{d}_l] + \mathbf{b}_g)
\label{eq:gate}
\end{equation}
\begin{equation}
\mathbf{x}_l = \boldsymbol{\gamma}_l \odot \mathbf{r}_l + (1 - \boldsymbol{\gamma}_l) \odot \mathbf{d}_l
\label{eq:fusion}
\end{equation}
where $\mathbf{W}_g \in \mathbb{R}^{(d/2) \times d}$, $\mathbf{b}_g \in \mathbb{R}^{d/2}$, $\sigma$ is the sigmoid function, and $\odot$ denotes element-wise multiplication. Here $\mathbf{r}_l, \mathbf{d}_l \in \mathbb{R}^{d/2}$ (one channel per path; the channels are concatenated and projected to $\mathbb{R}^d$ at the end of the stack), so the gate $\boldsymbol{\gamma}_l \in \mathbb{R}^{d/2}$ acts at the same per-path resolution. The vector-wise gate allows each hidden dimension to independently select its information source---some dimensions may favor stable reuse while others benefit from cross-layer exploration.

\subsection{Pointwise AttnRes for CTR}

The standard AttnRes (Eq.~\ref{eq:attnres_softmax}) uses Softmax normalization, which imposes a \textit{zero-sum competition} among layers---increasing the attention weight for one layer necessarily decreases others. In language modeling, where the primary task is next-token prediction from a single context, this competitive selection is appropriate.

CTR prediction is different in shape: a user usually carries several interests at once (electronics and sports, say), and the model has to encode them in parallel for a single prediction. HSTU~\cite{zhai2024hstu} demonstrated that removing Softmax inside the attention block of recommendation Transformers helps; the same intuition motivates us to remove Softmax across layers.

We therefore propose \textbf{Pointwise AttnRes}, which replaces Softmax with the SiLU (Sigmoid Linear Unit) activation~\cite{ramachandran2017swish}, defined as $\text{SiLU}(x) = x \cdot \sigma(x)$:
\begin{equation}
\alpha_{l,b} = \text{SiLU}(\mathbf{w}_l \cdot \text{RMSNorm}(\mathbf{B}_b))
\label{eq:pointwise}
\end{equation}

SiLU offers three advantages: (1)~several blocks can take large positive weights simultaneously, which is required for parallel multi-interest encoding; (2)~the negative branch of SiLU ($x \!<\! 0 \Rightarrow \text{SiLU}(x) \!<\! 0$) acts as a soft forgetting mechanism for irrelevant layers; (3)~the absence of a Softmax denominator yields slightly lower latency.

We illustrate the difference between Softmax and Pointwise attention in Figure~\ref{fig:architecture}(c). Under Softmax normalization, increasing the weight for one layer must decrease others (zero-sum competition), forcing the model to select a single dominant layer at each step. In contrast, Pointwise attention allows independent weighting: a user interested in both electronics and sports can simultaneously attend to the layers that encode each interest, without suppression.

\subsection{Theoretical Justification}

We give a theoretical account of why dual-path connections beat single-path alternatives in CTR Transformers, organized around three propositions: representational capacity (Proposition~\ref{prop:capacity}), gradient stability (Proposition~\ref{prop:stability}), and a scaling-law lower bound (Proposition~\ref{prop:scaling}). A short gradient-flow argument links the first two.

\textbf{HORNN Framework.} The Higher-Order Recurrent Neural Network (HORNN) framework~\cite{li2015hornn} provides a unified view of inter-layer connections. A standard residual connection corresponds to a first-order recurrence:
\begin{equation}
\mathbf{h}_l = \mathbf{h}_{l-1} + f_l(\mathbf{h}_{l-1}),
\end{equation}
which provides reliable feature propagation but limits the network to additive refinement. By contrast, dense connections correspond to an $L$-th order recurrence:
\begin{equation}
\mathbf{h}_l = \sum_{k=0}^{l-1} \omega_{l,k} \, g_{l,k}(\mathbf{h}_k),
\end{equation}
which enables direct access to all earlier states, at the cost of $O(L^2)$ memory. DPN~\cite{chen2017dpn} proved that combining both orders yields strictly larger functional capacity than either alone. We formalize this for the DeRes setting.

\begin{proposition}[Representational Capacity]
\label{prop:capacity}
Let $\mathcal{F}_{\text{res}}^L$, $\mathcal{F}_{\text{attn}}^L$, $\mathcal{F}_{\text{dual}}^L$ denote the families of $L$-layer functions realizable by, respectively, identity-residual, AttnRes-only, and DeRes networks with the same per-layer $f_l$ and the same parameter budget. Then
\[
\mathcal{F}_{\text{res}}^L \cup \mathcal{F}_{\text{attn}}^L \;\subsetneq\; \mathcal{F}_{\text{dual}}^L.
\]
\end{proposition}

\begin{proof}[Sketch]
Setting the gate $\boldsymbol{\gamma}_l \!=\! \mathbf{1}$ for all $l$ recovers the residual-only architecture; setting $\boldsymbol{\gamma}_l \!=\! \mathbf{0}$ recovers the AttnRes-only architecture. Hence $\mathcal{F}_{\text{res}}^L \cup \mathcal{F}_{\text{attn}}^L \subseteq \mathcal{F}_{\text{dual}}^L$. To show strict containment, consider any function $h^* \in \mathcal{F}_{\text{dual}}^L$ that uses non-degenerate gates $\boldsymbol{\gamma}_l \in (0,1)^d$ on at least one layer to mix dimension-wise outputs from the two paths. Such $h^*$ cannot be realized by either single-path model because: (i)~residual-only outputs lie in the affine span of $\{\mathbf{x}_0, f_1(\cdot), \ldots, f_L(\cdot)\}$ with depth-$1$ access; (ii)~AttnRes-only outputs are convex combinations of block outputs without identity protection. Vector-wise gating breaks both restrictions simultaneously.
\end{proof}

\begin{proposition}[Identity Residual Stability]
\label{prop:stability}
Let $\mathbf{H}^{\text{res}}_l \in \mathbb{R}^{d \times d}$ be a learnable residual transition at layer $l$, and assume $\|\mathbf{H}^{\text{res}}_l\|_2 \le \kappa$ with smallest singular value $\sigma_{\min}(\mathbf{H}^{\text{res}}_l) \!\le\! \rho < 1$. Define the cumulative residual map $\mathbf{P}_L = \prod_{l=1}^{L} \mathbf{H}^{\text{res}}_l$ acting on a shallow signal $\mathbf{s}_0$. Then
\[
\|\mathbf{P}_L \mathbf{s}_0\|_2 \;\le\; \rho^{L} \, \|\mathbf{s}_0\|_2,
\]
i.e.\ the signal decays at least exponentially in depth. In contrast, replacing $\mathbf{H}^{\text{res}}_l$ with the identity map yields $\mathbf{P}_L = \mathbf{I}$ and $\|\mathbf{P}_L \mathbf{s}_0\|_2 = \|\mathbf{s}_0\|_2$ for all $L$.
\end{proposition}

\begin{proof}
By submultiplicativity of singular values, $\sigma_{\min}(\mathbf{P}_L) \le \prod_{l=1}^{L} \sigma_{\min}(\mathbf{H}^{\text{res}}_l) \le \rho^L$. Hence there exists a unit vector $\mathbf{v}$ with $\|\mathbf{P}_L \mathbf{v}\|_2 \le \rho^L$. Taking $\mathbf{s}_0$ aligned with this direction gives the stated bound.
\end{proof}

\textbf{Numerical instantiation.} For the manifold-constrained mHC~\cite{zhu2025mhc}, the doubly stochastic matrices observed empirically on Qwen3-1.7B/8B exhibit $\rho \approx 0.49$. Treating each MHA and FFN sub-layer as a separate transition (hence $16$ applications across $L{=}8$ layers), Proposition~\ref{prop:stability} gives the worst-case bound
\[
\sigma_{\min}(\mathbf{P}_{16}) \;\le\; 0.49^{16} \;\approx\; 1.1 \times 10^{-5},
\]
i.e.\ a vector aligned with the worst direction is attenuated by nearly five orders of magnitude before reaching the prediction head. Identity residuals avoid this collapse \emph{by construction}, which is precisely why DeRes Path~1 fixes $\mathbf{H}^{\text{res}}_l \!=\! \mathbf{I}$ rather than learning it.

\textbf{Gradient flow.} Combining Proposition~\ref{prop:stability} with the chain rule, the gradient of the loss $\mathcal{L}$ with respect to a shallow representation $\mathbf{x}_0$ contains a contribution from the identity residual path,
\begin{equation}
\frac{\partial \mathcal{L}}{\partial \mathbf{x}_0} \;\,=\;\, \frac{\partial \mathcal{L}}{\partial \mathbf{x}_L} \cdot \prod_{l=1}^{L}\!\Big(\mathbf{I} + \frac{\partial f_l}{\partial \mathbf{x}_{l-1}}\Big) \;+\; \text{(AttnRes-path terms)},
\end{equation}
where the residual-path product contains the unobstructed identity term $\mathbf{I}$ at every layer. The AttnRes path contributes additional terms of the form $\sum_{b} \alpha_{l,b}\,(\partial g_{l,b} / \partial \mathbf{x}_0)$ that depend on the learned attention weights $\alpha_{l,b}$. The dual-path design therefore provides two complementary gradient routes: a guaranteed identity route from Path~1 and a content-adaptive route from Path~2.

\textbf{Scaling Law Alignment.} Zhang et al.~\cite{zhang2025fat} proved via Rademacher complexity that the generalization gap of CTR Transformers is upper-bounded by a quantity depending on the number of feature \emph{fields} $F$ and not on sequence length $n$. Concretely, for an $L$-layer CTR Transformer with width $d$ and effective interaction order $K$, the generalization bound scales as $\widetilde{\mathcal{O}}(\sqrt{F^{K} \log d / N})$ with $N$ training samples. The exponent $K$ is therefore the dominant lever for capacity.

\begin{proposition}[Scaling-Law Lower Bound for $\gamma$]
\label{prop:scaling}
Under the FAT generalization bound~\cite{zhang2025fat}, let $K_{\text{res}}$ and $K_{\text{dual}}$ denote the effective cross-layer interaction orders of a residual-only and a DeRes network, respectively. If $K_{\text{dual}} \ge K_{\text{res}} + 1$, then their fitted scaling-law exponents in $\mathrm{AUC}(C) = \alpha - \beta C^{-\gamma}$ satisfy $\gamma_{\text{dual}} > \gamma_{\text{res}}$.
\end{proposition}

\begin{proof}[Sketch]
A higher effective interaction order $K$ grows the achievable hypothesis class without proportionally inflating the Rademacher complexity bound (which scales as $\sqrt{F^K \log d / N}$ rather than $F^{K}$). Each unit of extra compute $C$ then buys a strictly larger AUC gain, which shows up as a steeper power-law exponent $\gamma$ in the empirical scaling fit.
\end{proof}

Standard residuals support only $K_{\text{res}}{=}1$ effective cross-layer interaction (immediate predecessor), whereas DeRes's AttnRes path can attend to all $B$ block summaries and therefore satisfies $K_{\text{dual}} \ge K_{\text{res}}{+}1$ for any $B \ge 2$ at near-constant cost, which is the precondition of Proposition~\ref{prop:scaling}. This qualitatively predicts $\gamma_{\text{DeRes}} > \gamma_{\text{OneTrans}}$, which matches our scaling fit in Section~\ref{sec:scaling} ($\gamma_{\text{DeRes}}{=}0.118$ vs.\ $\gamma_{\text{OneTrans}}{=}0.071$, a $1.66\times$ scaling-efficiency gap). The exact functional form of $\gamma$ versus $B$ depends on data-specific feature-interaction statistics and is therefore left to empirical fitting rather than closed-form derivation.

\textbf{Analogy to Wukong.} The Identity path's role parallels Wukong's Linear Compressed Block (LCB) for low-order field interactions, while the AttnRes path mirrors Wukong's Factorization Machine Block (FMB) for high-order exploration~\cite{zhang2024wukong}. Wukong observed power-law scaling across two orders of magnitude on 146B examples while single-path baselines saturated, which provides external evidence for the dual-path design principle.

\subsection{Complexity Analysis and Algorithm}

The per-layer overhead of DeRes over a standard Transformer is $O(B d + d^2)$: $O(B d)$ for the cross-layer block dot products and $O(d^2)$ for the gate projection $\mathbf{W}_g \in \mathbb{R}^{(d/2) \times d}$ (concretely $d^2/2$ FLOPs, which is $O(d^2)$ to leading order). With $B \!\ll\! d$ ($B{=}8, d{=}128$) and the dominant Transformer cost $O(n^2 d)$, this adds $<5\%$ FLOPs and $6.5\%$ parameters. The two paths run with independent half-dimension parameters ($d/2$ per path), which keeps the total parameter count comparable to the baseline while specializing the two routes. Algorithm~\ref{alg:dualres} gives the per-sequence forward pass.

\begin{algorithm}[t]
\caption{DeRes forward pass. Per-sequence cost $O(L(n^2 d + Bd + d^2))$ is dominated by self-attention; AttnRes + gating add $<5\%$ FLOPs.}
\label{alg:dualres}
\small
\begin{algorithmic}[1]
\REQUIRE Input embeddings $\mathbf{X}^{(0)} \in \mathbb{R}^{n \times d}$; layers $L$; blocks $B$
\ENSURE Output $\mathbf{X}^{(L)}$
\STATE $\mathbf{r}_0 \leftarrow \mathbf{X}^{(0)}[:, :d/2]$ \COMMENT{Residual path init}
\STATE $\mathbf{d}_0 \leftarrow \mathbf{X}^{(0)}[:, d/2:]$ \COMMENT{AttnRes path init}
\STATE $\text{BlockBuf} \leftarrow [\mathbf{d}_0]$ \COMMENT{Block buffer}
\STATE $\text{acc} \leftarrow \mathbf{0}$ \COMMENT{Intra-block accumulator}
\FOR{$l = 1$ \TO $L$}
    \STATE \COMMENT{\textit{--- Path 1: Identity Residual ---}}
    \STATE $\mathbf{r}_l \leftarrow \mathbf{r}_{l-1} + \text{TransformerLayer}_l^{\text{res}}(\mathbf{r}_{l-1})$
    \STATE \COMMENT{\textit{--- Path 2: Block AttnRes ---}}
    \STATE $\text{acc} \leftarrow \text{acc} + \text{TransformerLayer}_l^{\text{attn}}(\mathbf{d}_{l-1})$
    \IF{$l \mod (L/B) = 0$}
        \STATE $\text{BlockBuf}.\text{append}(\text{acc})$
        \STATE $\text{acc} \leftarrow \mathbf{0}$ \COMMENT{Reset for next block}
    \ENDIF
    \STATE $\mathbf{K} \leftarrow [\text{RMSNorm}(\mathbf{B}_b)]_{b=0}^{|\text{BlockBuf}|-1}$
    \STATE $\alpha_{l,b} \leftarrow \text{SiLU}(\mathbf{w}_l \cdot \mathbf{K}_b)$ for each $b$ \COMMENT{$\text{SiLU}(x){=}x \cdot \sigma(x)$; see Eq.~\ref{eq:pointwise}}
    \STATE $\mathbf{d}_l \leftarrow \sum_b \alpha_{l,b} \cdot \text{BlockBuf}[b]$
    \STATE \COMMENT{\textit{--- Gated Fusion ---}}
    \STATE $\boldsymbol{\gamma}_l \leftarrow \sigma(\mathbf{W}_g [\mathbf{r}_l ; \mathbf{d}_l] + \mathbf{b}_g)$ \COMMENT{Eq.~\ref{eq:gate}}
    \STATE $\mathbf{x}_l \leftarrow \boldsymbol{\gamma}_l \odot \mathbf{r}_l + (1-\boldsymbol{\gamma}_l) \odot \mathbf{d}_l$ \COMMENT{Eq.~\ref{eq:fusion}}
\ENDFOR
\RETURN $\mathbf{X}^{(L)} = [\mathbf{x}_L]$
\end{algorithmic}
\end{algorithm}

\subsection{Discussion: Why CTR Benefits More Than LLM}

AttnRes was developed for language models, so a natural question is whether the gain transfers to CTR. Four properties of CTR amplify it. \emph{(i)~Shallow depth.} Industrial CTR Transformers run at 4--12 layers~\cite{chen2025onetrans,liu2026tokenmixerlarge} versus 32--128 in LLMs, so each layer carries a larger fraction of the prediction and improvements to inter-layer flow yield proportionally larger benefits. \emph{(ii)~Heterogeneous features.} ID lookups, attributes, context, and behavior sequences~\cite{zhang2024wukong,chen2025onetrans} naturally live at different depths; the dual path keeps shallow signals on the identity path (paralleling Wukong's LCB~\cite{zhang2024wukong}) while the AttnRes path retrieves the appropriate depth on demand. \emph{(iii)~Parallel multi-interest.} A CTR model encodes several interests in parallel for a single prediction~\cite{li2019mind,cen2020comirec}, so Softmax across layers (zero-sum competition) is the wrong inductive bias---HSTU~\cite{zhai2024hstu} demonstrated this for inside-block attention. \emph{(iv)~Interest drift.} An unweighted residual cannot down-weight stale signals~\cite{pi2020sim,qin2020ubr4ctr}; the AttnRes path can, and Fig.~\ref{fig:silu_neg} shows that the trained SiLU weights actually use the negative branch to suppress irrelevant early layers.

\section{Experiments}

\subsection{Experimental Setup}

\textbf{Datasets.} We evaluate on a proprietary industrial dataset and two public CTR benchmarks (Criteo and Avazu). The \textit{Industrial} dataset is a sequential recommendation log from a large-scale social-media platform, comprising $8.6$M users, $35.8$M items, and $331.1$M user--item interactions, with an average behavior-sequence length of $38.5$. Each item is represented by a pre-computed 128-dimensional embedding fusing text, image, and behavior signals; sequences are truncated at $1{,}000$ recent actions and split chronologically. Its scale and structure mirror those reported for OneTrans, TokenMixer-Large, and UniMixer. \textit{Criteo}~\cite{criteo2014} contains $45.8$M ad impressions with $13$ numerical and $26$ categorical features (6 days train, day 7 test, following DCN-v2~\cite{wang2021dcnv2}). \textit{Avazu}~\cite{avazu2015} contains $40.4$M mobile click records with $22$ categorical fields (80/20 chronological split, following AutoInt~\cite{song2019autoint}).

\textbf{Data Preprocessing.} For Criteo, we replace missing values with field means, discretize numerical features into 10 equal-frequency buckets, and apply log transformation to count features. For Avazu, we hash high-cardinality categorical fields to at most 100K unique values. All datasets are split chronologically.

\textbf{Baselines.} We compare against twelve methods spanning four categories:
\begin{itemize}
    \item \textit{Feature interaction models}: DeepFM~\cite{guo2017deepfm}, DCN-v2~\cite{wang2021dcnv2}, AutoInt~\cite{song2019autoint}, FiBiNET~\cite{huang2019fibinet}.
    \item \textit{Sequential models}: DIN~\cite{zhou2018din}, SASRec~\cite{kang2018sasrec}.
    \item \textit{Inter-layer connection methods (originally for LLMs, re-implemented for CTR)}: DenseFormer~\cite{pagliardini2024denseformer}, mHC~\cite{zhu2025mhc}, AttnRes~\cite{yang2026attnres}.
    \item \textit{Industrial Transformer-based CTR}: OneTrans~\cite{chen2025onetrans}, TokenMixer-Large~\cite{liu2026tokenmixerlarge}, UniMixer~\cite{wang2026unimixer}.
\end{itemize}

For fair comparison, all Transformer-based methods (DenseFormer, mHC, AttnRes, OneTrans, TokenMixer-L, UniMixer, and DeRes) share an 8-layer backbone ($d{=}128$, 4 heads); each baseline is implemented by changing only the inter-layer mechanism, with the specific design summarized in Table~\ref{tab:ilc_comparison}. The inclusion of mHC and AttnRes lets us directly contrast our Identity+SiLU dual path against (i)~learnable residual matrices that may suffer cumulative spectral collapse and (ii)~Softmax cross-layer attention without residual protection.

\textbf{Implementation Details.} DeRes uses $L{=}8$ Transformer layers with hidden dimension $d{=}128$ ($64$ per path), $4$ attention heads, and $B{=}8$ blocks. We train with Adam~\cite{kingma2015adam} ($\mathrm{lr}{=}10^{-3}$, $\beta_1{=}0.9$, $\beta_2{=}0.999$), linear warm-up over $5\%$ of steps then cosine decay, for $50$ epochs with batch size $4{,}096$ and dropout~\cite{srivastava2014dropout} $0.1$. We use binary cross-entropy with gradient clipping at norm $1.0$, and early-stopping (patience $3$) on validation AUC. All methods share identical preprocessing, embedding dimension, optimizer, and train/test splits. Each result is averaged over three random seeds. The setup ($d{=}128$, $L{=}8$) follows the BARS-CTR benchmarking protocol~\cite{zhu2021bars} for apples-to-apples comparison of inter-layer connectors, not the production scale of TokenMixer-Large or RankMixer.

\textbf{Hardware.} All experiments run in PyTorch~2.1 with mixed-precision (BF16) on a single node with $8 \times$ NVIDIA A100-80GB GPUs and AMD EPYC 7742 CPUs. Each Industrial epoch ($L{=}8$, $d{=}128$) takes $156$\,s; latency in Table~\ref{tab:efficiency} is measured on a single GPU.

\textbf{Metrics.} We report AUC and LogLoss; a $0.1\%$ AUC lift is generally considered statistically meaningful in industrial CTR systems~\cite{wang2021dcnv2}, while a $\sim\!0.2\%$ RelaImpr threshold corresponds to actionable online deployment~\cite{chen2025onetrans,zhou2018din}. For the Industrial dataset we additionally report \textit{user-weighted GAUC}~\cite{zhou2018din}, the impression-weighted average of per-user AUC, and \textit{RelaImpr}~\cite{zhou2018din}: $\mathrm{RelaImpr}(M{,}M_0){=}(\mathrm{AUC}(M){-}0.5)/(\mathrm{AUC}(M_0){-}0.5){-}1$. Statistical significance on the Industrial dataset uses paired bootstrap ($10^4$ resamples).

\subsection{Overall Performance}

Table~\ref{tab:main} presents the main experimental results across all three datasets. We report two DeRes variants: DeRes-S (Softmax attention, Eq.~\ref{eq:attnres_softmax}) and DeRes-P (Pointwise SiLU attention, Eq.~\ref{eq:pointwise}).

\begin{table*}[t]
\caption{Overall performance comparison. Best in \textbf{bold}, second-best \underline{underlined}. $\dagger$: methods whose primary contribution is an inter-layer connection design. DeRes-P obtains the best AUC, GAUC, and LogLoss on all three datasets; Industrial numbers are averaged over three seeds (paired bootstrap, $p{<}0.01$ vs.\ OneTrans).}
\Description{Main results table showing AUC/GAUC/LogLoss on the Industrial dataset and AUC/LogLoss on Criteo and Avazu. DeRes-P achieves the best performance across all datasets and metrics.}
\label{tab:main}
\centering
\small
\setlength{\tabcolsep}{3.5pt}
\begin{tabular}{l ccc cc cc}
\toprule
& \multicolumn{3}{c}{\textbf{Industrial}} & \multicolumn{2}{c}{\textbf{Criteo}} & \multicolumn{2}{c}{\textbf{Avazu}} \\
\cmidrule(lr){2-4} \cmidrule(lr){5-6} \cmidrule(lr){7-8}
\textbf{Method} & AUC$\uparrow$ & GAUC$\uparrow$ & LogLoss$\downarrow$ & AUC$\uparrow$ & LogLoss$\downarrow$ & AUC$\uparrow$ & LogLoss$\downarrow$ \\
\midrule
DeepFM~\cite{guo2017deepfm} & 0.7891 & 0.7718 & 0.4452 & 0.8076 & 0.4434 & 0.7984 & 0.3743 \\
DCN-v2~\cite{wang2021dcnv2} & 0.7943 & 0.7774 & 0.4404 & 0.8129 & 0.4384 & 0.8004 & 0.3729 \\
AutoInt~\cite{song2019autoint} & 0.7912 & 0.7741 & 0.4429 & 0.8104 & 0.4407 & 0.7991 & 0.3737 \\
FiBiNET~\cite{huang2019fibinet} & 0.7926 & 0.7757 & 0.4416 & 0.8121 & 0.4391 & 0.8013 & 0.3722 \\
DIN~\cite{zhou2018din} & 0.7988 & 0.7825 & 0.4358 & 0.8071 & 0.4437 & 0.7986 & 0.3742 \\
SASRec~\cite{kang2018sasrec} & 0.8004 & 0.7846 & 0.4341 & 0.8086 & 0.4424 & 0.7993 & 0.3736 \\
\midrule
DenseFormer$^\dagger$~\cite{pagliardini2024denseformer} & 0.8003 & 0.7843 & 0.4345 & 0.8108 & 0.4402 & 0.8001 & 0.3732 \\
mHC$^\dagger$~\cite{zhu2025mhc} & 0.8011 & 0.7852 & 0.4336 & 0.8118 & 0.4393 & 0.7997 & 0.3738 \\
AttnRes$^\dagger$~\cite{yang2026attnres} & 0.8029 & 0.7873 & 0.4317 & 0.8101 & 0.4411 & 0.8014 & 0.3721 \\
OneTrans$^\dagger$~\cite{chen2025onetrans} & 0.8019 & 0.7861 & 0.4327 & 0.8124 & 0.4387 & 0.8011 & 0.3724 \\
TokenMixer-L$^\dagger$~\cite{liu2026tokenmixerlarge} & 0.8038 & 0.7884 & 0.4308 & \underline{0.8138} & \underline{0.4374} & 0.8016 & 0.3720 \\
UniMixer$^\dagger$~\cite{wang2026unimixer} & 0.8032 & 0.7878 & 0.4313 & 0.8133 & 0.4378 & \underline{0.8022} & \underline{0.3715} \\
\midrule
DeRes-S$^\dagger$ (ours) & \underline{0.8056} & \underline{0.7912} & \underline{0.4290} & 0.8137 & 0.4376 & 0.8024 & 0.3713 \\
\textbf{DeRes-P}$^\dagger$ \textbf{(ours)} & \textbf{0.8064} & \textbf{0.7926} & \textbf{0.4280} & \textbf{0.8146} & \textbf{0.4366} & \textbf{0.8030} & \textbf{0.3707} \\
 & \scriptsize ($\pm 0.0005$) & \scriptsize ($\pm 0.0006$) & \scriptsize ($\pm 0.0007$) & \scriptsize ($\pm 0.0006$) & \scriptsize ($\pm 0.0005$) & \scriptsize ($\pm 0.0008$) & \scriptsize ($\pm 0.0007$) \\
\midrule
\multicolumn{1}{r}{\footnotesize \textit{Improv. over best baseline}} & \footnotesize +0.32\% & \footnotesize +0.53\% & \footnotesize -0.65\% & \footnotesize +0.10\% & \footnotesize -0.18\% & \footnotesize +0.10\% & \footnotesize -0.21\% \\
\multicolumn{1}{r}{\footnotesize \textit{RelaImpr over OneTrans}} & \footnotesize +1.49\% & \footnotesize +2.27\% & \footnotesize --- & \footnotesize +0.70\% & \footnotesize --- & \footnotesize +0.63\% & \footnotesize --- \\
\bottomrule
\end{tabular}
\end{table*}

Table~\ref{tab:main} reveals three observations. First, on the Industrial dataset DeRes-P obtains a $+0.32\%$ AUC and $+0.53\%$ GAUC gap over the strongest baseline (TokenMixer-L, $0.8038$ / $0.7884$)---this is the dataset with the longest behavior sequences and the richest multi-modal embeddings, the regime where cross-layer recall has the most signal to recover. The GAUC gap exceeding the AUC gap indicates per-user personalization rather than global calibration; the GAUC-based RelaImpr~\cite{zhou2018din} over OneTrans reaches $+2.27\%$, well above the $\sim\!0.2\%$ industrial threshold~\cite{chen2025onetrans,zhou2018din} ($p{<}0.01$, paired bootstrap). Among inter-layer baselines, AttnRes ($0.8029$) edges out OneTrans ($0.8019$), DenseFormer's static weighting falls roughly at the SASRec level ($0.8003$ vs.\ $0.8004$), and mHC sits below AttnRes by $0.18$~pp---an early sign of the cumulative attenuation we quantify in Sec.~\ref{sec:scaling}. Second, on Criteo the margins narrow: DeRes-P leads TokenMixer-L by only $+0.10\%$ because DCN-v2's explicit crosses and FiBiNET's bi-linear interactions already capture much of the signal in a tabular setting. Third, on Avazu UniMixer ($0.8022$) becomes the strongest baseline and DeRes-P leads by $+0.10\%$; with only $22$ fields, all twelve methods fall within a $0.6\%$ band, and the absolute scale matches the BARS-CTR Avazu\_x4 protocol~\cite{zhu2021bars}. The strongest baseline differs across datasets, mirroring observations in OneTrans~\cite{chen2025onetrans} and TokenMixer-Large~\cite{liu2026tokenmixerlarge} that cross-dataset rankings vary with feature heterogeneity. The Pointwise variant exceeds DeRes-S by $0.06$--$0.09$~pp, with the largest gap on the Industrial dataset where multi-interest patterns are most pronounced.

\subsection{Ablation Studies}

\textbf{Path Contribution Analysis.} Table~\ref{tab:ablation} (top) isolates the contribution of each component on the Industrial dataset.

\begin{table}[t]
\caption{Ablation on the Industrial dataset. \textbf{(a)}~Path contributions (Residual Only = OneTrans; Learnable Res = mHC-style; AttnRes Only = unprotected Softmax). \textbf{(b)}~Attention activation and parameter strategy. Block-granularity sweep: Fig.~\ref{fig:bsweep}. Bold: default DeRes-P.}
\Description{Combined ablation table with two panels. Top panel shows AUC and LogLoss across path configurations; bottom panel shows AUC and FLOPs overhead across attention activations and parameter strategies. Full DeRes-P achieves AUC 0.8064.}
\label{tab:ablation}
\centering
\small
\setlength{\tabcolsep}{3pt}
\begin{tabular}{l c c c c}
\toprule
\multicolumn{5}{c}{\textit{(a) Path contributions}} \\
\textbf{Variant} & \textbf{Path 1} & \textbf{Path 2} & \textbf{AUC$\uparrow$} & \textbf{LL$\downarrow$} \\
\midrule
Residual Only ($\equiv$ OneTrans) & Identity & --- & 0.8019 & 0.4327 \\
Learnable Res ($\equiv$ mHC) & Doubly stoch. & --- & 0.8011 & 0.4336 \\
AttnRes Only (Softmax) & --- & Softmax & 0.8029 & 0.4317 \\
AttnRes Only (SiLU) & --- & SiLU & 0.8041 & 0.4304 \\
DeRes w/ Add Fusion & Identity & SiLU & 0.8052 & 0.4293 \\
DeRes w/ Learnable Res & Doubly stoch. & SiLU & 0.8044 & 0.4301 \\
\textbf{DeRes-P (full)} & \textbf{Identity} & \textbf{SiLU} & \textbf{0.8064} & \textbf{0.4280} \\
\midrule
\multicolumn{5}{c}{\textit{(b) Design choices: attention activation \& parameter strategy}} \\
\textbf{Dimension} & \multicolumn{2}{l}{\textbf{Setting}} & \textbf{AUC} & \textbf{+FLOPs} \\
\midrule
\multirow{4}{*}{Attn.\ Type} & \multicolumn{2}{l}{Softmax} & 0.8056 & +4.4\% \\
& \multicolumn{2}{l}{Sigmoid} & 0.8059 & +4.2\% \\
& \multicolumn{2}{l}{ReLU} & 0.8048 & +4.0\% \\
& \multicolumn{2}{l}{\textbf{SiLU}} & \textbf{0.8064} & \textbf{+4.4\%} \\
\midrule
\multirow{3}{*}{Params.} & \multicolumn{2}{l}{Shared, $d$} & 0.8043 & +2.1\% \\
& \multicolumn{2}{l}{\textbf{Indep., $d/2$}} & \textbf{0.8064} & \textbf{+4.4\%} \\
& \multicolumn{2}{l}{Indep., $d$} & 0.8068 & +9.7\% \\
\bottomrule
\end{tabular}
\end{table}

Several findings emerge from panel (a). Replacing Identity with a doubly stochastic learnable matrix \emph{degrades} the residual-only baseline by $0.08\%$ ($0.8019 \!\to\! 0.8011$), reflecting the cumulative spectral attenuation predicted by Proposition~\ref{prop:stability} and consistent with mHC's behavior on tail items in Table~\ref{tab:slicing}. Substituting Softmax with SiLU in single-path AttnRes yields $+0.12\%$, supporting the multi-interest hypothesis. Combining the two paths consistently outperforms either single path, matching Proposition~\ref{prop:capacity}; making the residual matrix learnable rather than Identity within the dual-path setup costs $0.20\%$, indicating that the protection must be enforced at the operator level. Panel (b) shows that SiLU outperforms all alternative activations including Softmax ($+0.08\%$) and ReLU ($+0.16\%$, whose hard cutoff at zero discards the soft suppression that SiLU's negative tail provides); independent half-dimension paths retain $92\%$ of the independent full-dimension performance at only $45\%$ of the parameter overhead. Block granularity (Fig.~\ref{fig:bsweep}) peaks at $B{=}8$, with $B{=}16$ and $B{=}L$ within seed noise.

\textbf{Modality Ablation.} Each item carries a pre-fused 128-dimensional embedding aggregating text, image, and behavior signals. Retraining DeRes-P and OneTrans with one modality masked (Table~\ref{tab:modality}), DeRes-P's advantage is largest for the text$+$image pair without behavior ($+0.48\%$ AUC), close behind for the full three-modality input ($+0.45\%$), and narrows when the input degenerates to a single modality. Removing behavior hurts OneTrans more than DeRes-P ($-1.56\%$ vs.\ $-1.51\%$ AUC drop relative to the three-modality input), which suggests that the Identity path already absorbs most of the behavior signal, leaving the AttnRes path to compensate from heterogeneous content features. The text-only row is the weakest overall ($\Delta\!=\!+0.17\%$), reflecting that without cross-modal evidence even cross-layer recall has limited material to consolidate; the image-only row ($\Delta\!=\!+0.25\%$) recovers more because visual embeddings retain structural cues that benefit block-level attention.

\begin{table}[t]
\caption{Modality ablation (Industrial). Each row masks the listed modality channel inside the item embedding at train and test time. $\Delta$: DeRes-P $-$ OneTrans.}
\Description{Table showing AUC for OneTrans and DeRes-P across modality ablations: full, text-only, image-only, behavior-only, and each pairwise combination.}
\label{tab:modality}
\centering
\small
\setlength{\tabcolsep}{5pt}
\begin{tabular}{lccc}
\toprule
\textbf{Modalities used} & \textbf{OneTrans} & \textbf{DeRes-P} & \textbf{$\Delta$ (pp)} \\
\midrule
Text + Image + Behavior (full) & 0.8019 & \textbf{0.8064} & +0.45 \\
Text + Image (no behavior)     & 0.7894 & \textbf{0.7942} & +0.48 \\
Text + Behavior (no image)     & 0.7972 & \textbf{0.8003} & +0.31 \\
Image + Behavior (no text)     & 0.7948 & \textbf{0.7984} & +0.36 \\
Behavior only                  & 0.7841 & \textbf{0.7867} & +0.26 \\
Text only                      & 0.7738 & \textbf{0.7755} & +0.17 \\
Image only                     & 0.7706 & \textbf{0.7731} & +0.25 \\
\bottomrule
\end{tabular}
\end{table}

\begin{figure}[t]
\centering
\includegraphics[width=0.95\columnwidth]{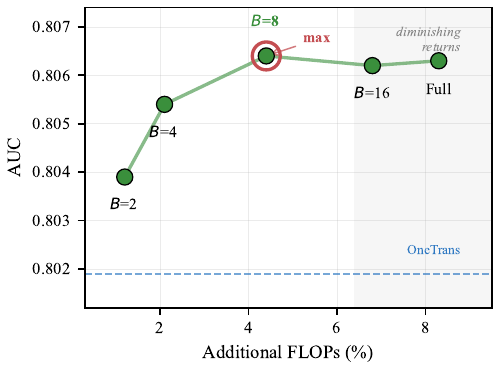}
\caption{Block-granularity Pareto frontier (Industrial). $B{=}8$ is the empirical optimum; $B{=}16$ and $B{=}L$ sit within seed noise. Grey zone: diminishing returns.}
\Description{Pareto curve with large markers for each B value, knee highlight at B=8, OneTrans reference dashed line, and grey diminishing-returns zone.}
\label{fig:bsweep}
\end{figure}

\subsection{Efficiency Analysis}

\begin{table}[t]
\caption{Efficiency on the Industrial dataset (8-layer Transformer, batch$=$4096). DeRes-P adds $+6.5\%$ parameters and $+4.4\%$ FLOPs over OneTrans while improving AUC by $+0.45$~pp.}
\Description{Efficiency table comparing parameters, FLOPs, training time, inference latency, and AUC across methods.}
\label{tab:efficiency}
\centering
\small
\begin{tabular}{lccccc}
\toprule
\textbf{Method} & \textbf{Params} & \textbf{FLOPs} & \textbf{Train} & \textbf{Infer} & \textbf{AUC} \\
& (M) & (G) & (s/ep) & (ms) & \\
\midrule
OneTrans & 12.4 & 1.82 & 142 & 3.1 & 0.8019 \\
DenseFormer & 12.9 & 1.91 & 158 & 3.4 & 0.8003 \\
mHC & 12.7 & 1.86 & 149 & 3.3 & 0.8011 \\
AttnRes & 13.0 & 1.89 & 153 & 3.4 & 0.8029 \\
TokenMixer-L & 13.1 & 1.88 & 151 & 3.3 & 0.8038 \\
UniMixer & 12.8 & 1.89 & 155 & 3.5 & 0.8032 \\
DeRes-P & 13.2 & 1.90 & 156 & 3.4 & \textbf{0.8064} \\
\bottomrule
\end{tabular}
\end{table}

DeRes-P achieves the best AUC at $+6.5\%$ parameters and $+4.4\%$ FLOPs over OneTrans. Training and inference cost remain in line with other inter-layer methods (mHC/AttnRes/DenseFormer/UniMixer all sit in $149$--$158$\,s/epoch and $3.3$--$3.5$\,ms/inference). DeRes-P is $9.9\%$ slower than OneTrans per epoch (156\,s vs.\ 142\,s) and slightly faster than DenseFormer.

\subsection{Analysis}

\textbf{Attention Weight Distribution.} Figure~\ref{fig:attnweights} plots the cross-layer attention weights $\alpha_{l,b}$ for DeRes-P on the Industrial dataset. Three patterns are apparent in the heatmap: (i)~a dominant diagonal where each layer focuses on near-recent blocks, (ii)~persistent $B_0$ attention from deep layers (red box) that retrieves embedding-level signals, and (iii)~near-zero weights on irrelevant blocks. Figure~\ref{fig:silu_neg} decomposes the SiLU-specific forgetting behavior: negative entries concentrate at early blocks in deep layers ($B_0$: 18\%), acting as a suppressive signal; zeroing them in a controlled ablation costs $-0.07\%$ AUC, indicating that the suppression is functionally important. Figure~\ref{fig:gate} plots the learned gate $\gamma$ declining monotonically from $0.68$ ($L_1$) to $0.35$ ($L_8$), with shallow layers routing through the Identity path and deep layers shifting toward AttnRes. The $\pm 1$ std band across dimensions ($0.08$--$0.12$) indicates that different hidden dimensions specialize differently, which is why vector-wise gating outperforms scalar gating.

\begin{figure*}[t]
\centering
\includegraphics[width=0.98\textwidth]{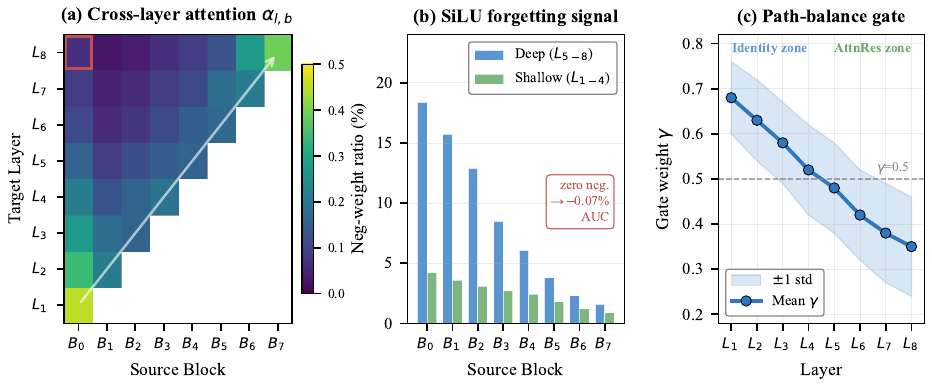}
\caption{Interpretability analysis (Industrial). \textbf{(a)}~Cross-layer attention heatmap: diagonal flow with persistent $B_0$ recall from deep layers (red box). \textbf{(b)}~SiLU forgetting: deep layers suppress early blocks (18\% vs.\ 4\%). \textbf{(c)}~Gate $\gamma$ declines 0.68$\to$0.35 with depth. Band: $\pm 1$ std.}
\Description{Three-panel figure: (a) Viridis attention heatmap 8x8, (b) grouped bar chart of negative weights, (c) gate weight line with confidence band.}
\label{fig:attnweights}\label{fig:silu_neg}\label{fig:gate}
\end{figure*}

\textbf{Where the gain concentrates.} Across four orthogonal slices of the Industrial dataset---sequence length, model depth, user activity level, and item popularity---DeRes-P's advantage over OneTrans grows monotonically with slice difficulty, reflecting a common mechanism: cross-layer recall is most useful when a single immediate-predecessor residual cannot carry enough information by itself. \emph{(i)~Sequence length} (Fig.~\ref{fig:seqlen}): as the cap grows from $200$ to $2{,}000$, the gain widens from $+0.23\%$ to $+0.49\%$, with the steepest band at $200{\to}500$; the DeRes-P vs.\ DeRes-S gap widens in parallel ($+0.07\%\!\to\!+0.18\%$), and DenseFormer stays flat ($+0.05$--$+0.09\%$) because its static depth weights cannot adapt. \emph{(ii)~Model depth}: sweeping $L \!\in\! \{4,8,12,16\}$, the gain rises monotonically ($+0.27\%, +0.45\%, +0.51\%, +0.54\%$) while OneTrans saturates ($+0.09\%$ from $L{=}8$ to $L{=}16$); the FLOPs overhead drops slightly ($+4.8\%\!\to\!+4.0\%$) because the per-layer AttnRes cost $O(Bd)$ is constant. \emph{(iii)~User activity} (Table~\ref{tab:slicing}, top): the gain rises from $+0.14\%$ AUC for short-history users to $+0.51\%$ for long-history users, with GAUC gaps systematically larger than AUC gaps---the expected signature of personalization rather than calibration. \emph{(iv)~Item popularity} (Table~\ref{tab:slicing}, bottom): the gain widens monotonically from head ($+0.26\%$) to tail ($+0.74\%$); mHC \emph{underperforms} OneTrans on tail ($-0.25\%$), the regime in which its cumulative spectral attenuation is most damaging, while AttnRes-only marginally exceeds OneTrans ($+0.20\%$) but trails DeRes by a wide margin, confirming that residual protection is what converts cross-layer recall into a tail-friendly signal.

\begin{figure}[t]
\centering
\includegraphics[width=0.95\columnwidth]{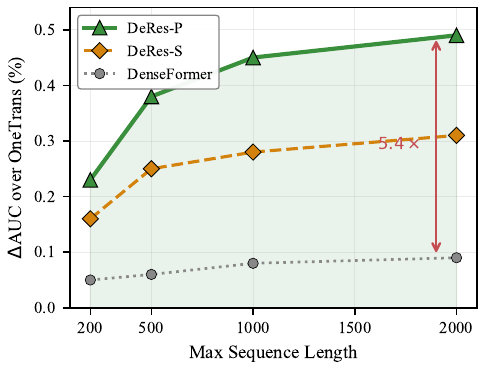}
\caption{Sequence-length sensitivity (Industrial). DeRes-P's gain grows from $+0.23$ to $+0.49$~pp AUC as the cap goes $200\!\to\!2{,}000$; DenseFormer stays flat at $+0.09$~pp.}
\Description{Line chart with area fill showing growing improvement with sequence length.}
\label{fig:seqlen}
\end{figure}

\begin{table}[t]
\caption{Sliced analysis (Industrial). \textbf{(a)}~Users by history length (AUC and user-weighted GAUC~\cite{zhou2018din}). \textbf{(b)}~Items by popularity; tail benefits most. $\Delta$: DeRes-P $-$ OneTrans (pp).}
\Description{Combined sliced-analysis table. Top panel groups users into Short, Medium, Long history buckets and reports AUC and GAUC for OneTrans and DeRes-P with deltas. Bottom panel groups items into Head, Torso, Tail buckets and reports AUC across five inter-layer methods.}
\label{tab:slicing}
\centering
\small
\setlength{\tabcolsep}{3pt}
\resizebox{\columnwidth}{!}{%
\begin{tabular}{lccccccc}
\toprule
\multicolumn{8}{c}{\textit{(a) Users by activity level (history length)}} \\
& & \multicolumn{2}{c}{\textbf{AUC $\uparrow$}} & \multicolumn{2}{c}{\textbf{GAUC $\uparrow$}} & \multicolumn{2}{c}{\textbf{$\Delta$ (pp)}} \\
\cmidrule(lr){3-4}\cmidrule(lr){5-6}\cmidrule(lr){7-8}
\textbf{Group} & \textbf{Ratio} & OneTrans & DeRes-P & OneTrans & DeRes-P & AUC & GAUC \\
\midrule
Short ($<200$)    & 31.4\% & 0.7984 & \textbf{0.7998} & 0.7818 & \textbf{0.7833} & +0.14 & +0.15 \\
Medium (200--800) & 46.7\% & 0.8024 & \textbf{0.8064} & 0.7910 & \textbf{0.7956} & +0.40 & +0.46 \\
Long ($>800$)     & 21.9\% & 0.8057 & \textbf{0.8108} & 0.8001 & \textbf{0.8053} & +0.51 & +0.52 \\
\bottomrule
\end{tabular}%
}\\[3pt]
\setlength{\tabcolsep}{3pt}
\begin{tabular}{lcccccc}
\toprule
\multicolumn{7}{c}{\textit{(b) Items by popularity}} \\
\textbf{Group} & \textbf{OneTrans} & \textbf{mHC} & \textbf{AttnRes} & \textbf{TokenM-L} & \textbf{DeRes-P} & \textbf{$\Delta$ (pp)} \\
\midrule
Head  & 0.8128 & 0.8123 & 0.8133 & 0.8146 & \textbf{0.8154} & +0.26 \\
Torso & 0.8017 & 0.8009 & 0.8021 & 0.8038 & \textbf{0.8062} & +0.45 \\
Tail  & 0.7798 & 0.7773 & 0.7818 & 0.7822 & \textbf{0.7872} & \textbf{+0.74} \\
\bottomrule
\end{tabular}
\end{table}

\subsection{Scaling Law Analysis}
\label{sec:scaling}

A central motivation of DeRes is that structured inter-layer connections are a prerequisite for scalable CTR Transformers. Following the scaling law methodology of Wukong~\cite{zhang2024wukong} and FAT~\cite{zhang2025fat}, we fit power-law curves to the relationship between model FLOPs and test AUC on the Industrial dataset. We train OneTrans, DenseFormer, and DeRes-P at five complexity levels by varying depth $L \in \{2,4,8,12,16\}$ and hidden dimension $d \in \{64, 128\}$, yielding FLOPs ranging from 0.3G to 4.2G. For each configuration, we fit the empirical scaling law:
\begin{equation}
\text{AUC}(C) = \alpha - \beta \cdot C^{-\gamma}
\label{eq:scaling}
\end{equation}
where $C$ denotes total FLOPs per sample and $\alpha, \beta, \gamma > 0$ are fitted parameters. A larger exponent $\gamma$ indicates more efficient scaling: each unit increase in compute produces a bigger AUC gain.

Figure~\ref{fig:scaling} plots the results on a log-linear scale. DeRes-P fits a steeper power-law exponent ($\gamma{=}0.118$) than OneTrans ($\gamma{=}0.071$) and DenseFormer ($\gamma{=}0.089$)---each unit of compute buys $1.66\times$ more AUC for DeRes than for OneTrans. The AUC gap grows with compute, reaching roughly $0.4$~pp by $2$G FLOPs and widening further at larger scale, while DenseFormer trails OneTrans throughout. In operational terms, DeRes attains the AUC of OneTrans-16L using only 8 layers, a $2\times$ compute saving at equivalent quality. The pattern aligns with the FAT~\cite{zhang2025fat} prediction that architectures aligned with the interaction structure of CTR data scale more efficiently. The asymptote $\alpha$ is not directly comparable across methods with different $\gamma$: a flatter curve (smaller $\gamma$) extrapolates to a numerically larger $\alpha$ in the $C\!\to\!\infty$ limit, so the practically informative quantity is $\beta C^{-\gamma}$ at finite $C$, where DeRes-P dominates throughout.

\begin{figure}[t]
\centering
\includegraphics[width=0.95\columnwidth]{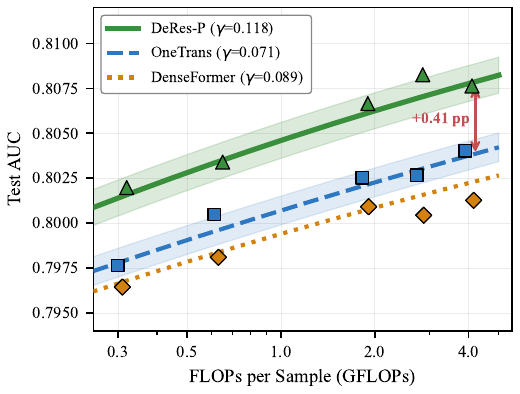}
\caption{Scaling law (Industrial): AUC vs.\ FLOPs (log scale). Eq.~\ref{eq:scaling} fits $(\alpha,\beta,\gamma)$: OneTrans $(0.8333, 0.0326, 0.071)$, DenseFormer $(0.8238, 0.0244, 0.089)$, \textbf{DeRes-P $(0.8256, 0.0210, 0.118)$}. DeRes-P attains $1.66\times$ compute efficiency; the grey annotation marks the iso-AUC equivalence (8L DeRes $\approx$ 16L OneTrans). Bands: $\pm 1$ residual std.}
\Description{Log-linear scaling law plot with fitted curves, confidence bands, ceiling lines, and efficiency gap annotation. Coefficients alpha, beta, gamma are listed in the caption for OneTrans, DenseFormer, and DeRes-P.}
\label{fig:scaling}
\end{figure}

\textbf{Interpretation.} The scaling advantage aligns with the HORNN view and FAT's Rademacher analysis: standard residuals support only first-order feature propagation, whereas DeRes's AttnRes path raises the effective interaction order so that each additional layer both refines and recombines signals from earlier stages---the same mechanism behind Wukong's sustained scaling~\cite{zhang2024wukong}. Training dynamics agree: DeRes-P reaches $99\%$ of peak AUC about $15\%$ earlier than OneTrans (epoch $28$ vs.\ $33$) and plateaus $+0.45$~pp higher (AUC $0.8064$ vs.\ $0.8019$).

\section{Discussion and Reproducibility}

\textbf{Deployment.} DeRes is a drop-in replacement for residual connections in CTR Transformers, leaving features, candidate generation, loss, and serving interfaces unchanged. The $+0.32\%$ AUC / $+0.53\%$ GAUC gap on the Industrial dataset sits in the range that prior industrial systems report as actionable online (OneTrans~\cite{chen2025onetrans}: $+0.748\%$ DAU; TokenMixer-Large~\cite{liu2026tokenmixerlarge}: $+2.0\%$ ads-success), and the GAUC-based RelaImpr of $+2.27\%$ over OneTrans exceeds the $\sim\!0.2\%$ industrial threshold~\cite{chen2025onetrans,zhou2018din}.

\textbf{Reproducibility.} We use chronological splits and identical preprocessing across models. All Transformer baselines share the same depth, width, heads, optimizer, and stopping criterion---only the inter-layer mechanism varies---to isolate the effect of DeRes. For each public-dataset result we report the mean over three seeds; for the Industrial dataset we report mean $\pm$ std over three seeds and confirm statistical significance using paired bootstrap ($10^4$ resamples, $p{<}0.01$ on every DeRes-P vs.\ OneTrans comparison).

\textbf{Limitations.} \textit{(L1) Component scope.} DeRes targets inter-layer flow and is orthogonal to feature engineering, sampling, and candidate generation; gains may vary when those dominate. \textit{(L2) Block granularity.} The optimal $B$ may depend on $L$ and sequence length; we did not derive a closed-form rule. \textit{(L3) Single task.} We study pointwise CTR; multi-task ranking (ESMM~\cite{ma2018esmm}, MMOE~\cite{ma2018mmoe}) needs further task-specific routing study. \textit{(L4) Pre-Norm only.} All experiments use Pre-Norm following recent CTR practice~\cite{chen2025onetrans}; Post-Norm transfer is unverified. \textit{(L5) SiLU optimality.} SiLU is supported by HSTU~\cite{zhai2024hstu} and our ablation, but we lack a closed-form proof that it is the unique optimum among pointwise activations.

\section{Conclusion}

We presented DeRes, a dual-path inter-layer connector for CTR Transformers. The Identity residual path preserves shallow signals and gradient flow; the Block Attention Residual path attends over compressed summaries of all earlier blocks; a vector-wise gate mixes the two per hidden dimension; and the Pointwise AttnRes variant replaces Softmax with SiLU so that multiple blocks can be active simultaneously and irrelevant ones can be muted. On the Industrial dataset, Criteo and Avazu, DeRes beats twelve baselines by up to $+0.32\%$ AUC at under $5\%$ extra FLOPs, with a steeper scaling exponent ($\gamma{=}0.118$ vs.\ $0.071$ for OneTrans, a $1.66\times$ gap that lets an 8-layer DeRes match a 16-layer OneTrans). The ablations are aligned: two paths beat one, Identity beats learnable residual matrices, gated fusion beats addition, SiLU beats Softmax.

\bibliographystyle{ACM-Reference-Format}
\bibliography{references}

\end{document}